\documentclass[10pt, conference]{IEEEtran}

\makeatletter
\def\ps@headings{%
\def\@oddhead{\mbox{}\scriptsize\rightmark \hfil \thepage}%
\def\@evenhead{\scriptsize\thepage \hfil \leftmark\mbox{}}%
\def\@oddfoot{}%
\def\@evenfoot{}}
\makeatother

\newcommand{\nop}[1]{}
\usepackage{stmaryrd}
\usepackage{bbding}
\usepackage{algorithmic}
\usepackage{algorithm}
\usepackage{graphicx}
\usepackage{epsfig}
\usepackage{subfigure}
\usepackage{multirow}
\usepackage{balance}

\usepackage[centertags]{amsmath}

\newtheorem{definition}{Definition}
\newtheorem{remark}{Remark}

\newtheorem{lemma}{Lemma}
\newtheorem{theorem}{Theorem}

\newtheorem{proof}{Proof}

\def\linespaces{0.95}
\def\baselinestretch{\linespaces}

\nop{

\documentclass[11pt,conference]{IEEEtran}

\usepackage{graphicx}
\usepackage{balance}
\usepackage{amsmath}
\begin{document}
\newcommand{\qed}{\mbox{\rule[0pt]{1.0ex}{1.0ex}}}
\def\boxend{\hspace*{\fill} $\QED$}

\newtheorem{definition}{Definition}
\newtheorem{example}{Example}
\newtheorem{theorem}{Theorem}
\newtheorem{problem}{Problem}
\newtheorem{lemma}{Lemma}
\newtheorem{proposition}{Proposition}
\newtheorem{remark}{Remark}

}

\newcounter{line}
\input epsf

\begin{document}

\title{On the Model Transform in Stochastic Network Calculus}

\author{\IEEEauthorblockN{Kui Wu\IEEEauthorrefmark{1},
Yuming Jiang\IEEEauthorrefmark{2}, and Jie Li\IEEEauthorrefmark{3}}\\
}


\maketitle

\begin{abstract}
Stochastic network calculus requires special care in the search of proper stochastic traffic arrival models and stochastic service models. Tradeoff must be considered between the feasibility for the analysis of performance bounds, the usefulness of performance bounds, and the ease of their numerical calculation. In theory, transform between different traffic arrival models and transform between different service models are possible. Nevertheless, the impact of the model transform on performance bounds has not been thoroughly investigated. This paper is to investigate the effect of the model transform and to provide practical guidance in the model selection in stochastic network calculus.   
\end{abstract}




\section{Introduction}\label{sec:introduction}

Performance has always been one of the major concerns in networking systems. Mathematical models for quantitative evaluation of network performance, however, have remained as a slow-paced research area. A.K. Erlang published the first paper on queuing
theory in 1909~\cite{Erlang}, and since then queuing theory has been developed and applied in a wide
variety of applications. In particular, it has been the foundation in performance modeling and
evaluation of telecommunication systems and has been applied broadly
in the performance analysis of computer networks. Nevertheless, with
the advance of the Internet technology, the assumptions behind the
tractable queuing models may not hold anymore. Despite the research efforts in the last one hundred years, the tractable models with traditional queueing theory consist of only a minority of practical network problems. The research community is in dire need of new mathematical models for network-wide performance evaluation where the Markovian property in traffic arrivals or services may not hold.

Network calculus is one of such new analytical techniques. The theory of network calculus was introduced in early
1990s~\cite{Cruz91a} for network performance evaluation. Unlike the traditional queueing theory which aims at obtaining exact analytical results, network calculus focuses on the analysis of performance \textit{bounds} using the cumulative amount of traffic arrivals or services. Since network calculus usually does not assume particular distributions on traffic arrivals or service times, it can obtain broadly applicable performance results. Network calculus has evolved along two tracks-- deterministic~\cite{Chang00,Le}
and stochastic~\cite{Jia,Kurose92,Li,Sidi93}. The deterministic network calculus is to obtain the \textit{worst} case performance bounds,
which may be too loose to be useful in practice. Due to this reason, research on this direction gradually fades out. To overcome the
problem, stochastic network calculus was developed. Nevertheless, due to some special
difficulties~\cite{Jia,Li}, basic properties of stochastic network calculus have been
proved only in recent years~\cite{Ciucu06,Jia,Jiangbook}. Although the major theoretical barriers have been cleared, it is still unclear whether or not stochastic network calculus will be broadly accepted by network practitioners.  

\nop{It is fair to say that stochastic network calculus is still confined to a small coterie of researchers.} Without the driving force from real applications, broad acceptance of stochastic network calculus as a valuable technique for performance evaluation may not be optimistic. One of the major practical challenges is the lack of effective algorithms to calculate and compare the performance bounds. After all, what really matter to network engineers are the meaningful numerical results instead of the complex equations. Although there are some efforts using Legendre transform~\cite{Fidler05,Jia3} to simplify the calculation of major operations in network calculus, the treatments are far from sufficient to tackle the difficulties in the stochastic network calculus, where we are often faced with multiple tradeoffs. 

Specifically, three tradeoffs must be considered in stochastic network calculus. First, there is a tradeoff between the simplicity of deriving performance bounds and the difficulty in the numerical calculation of the bounds. It is known that in order to derive performance bounds easily, we need to put extra constraints on the traffic and the service models~\cite{Jiangbook}. For instance, we may need to put some constraints in the cumulative traffic arrivals/services, e.g, we may change the calculation from the form of $Prob\{f_t > 0\}$ to the form of $Prob\{\sup_{t} f_t >0\}$, which is usually not equal to $\sup_{t} Prob\{f_t >0\}$. Note that $sup$ is the supremum (i.e., least upper bound) operation. $Prob\{\sup_{t} f_t >0\}$ thus represents an instantaneous property and is generally hard to calculate. Second, there is a tradeoff between the usefulness of the traffic (service) models and the hardness of searching for these models in real applications. This tradeoff is closely related to the first one. In general, it is easy to obtain the traffic model (or service model) directly from the distribution of packet inter-arrival times (or the distribution of the service times). Introducing extra constraints on the traffic arrival or service model, e.g., adding the $sup$ operation in the model~\cite{Jia}, requires that we either perform model transform~\cite{Jiangbook} or search for the models using queueing analysis methods~\cite{Fry,Kingman}. Third, we must consider the tightness of performance bounds in a stochastic sense. In stochastic network calculus, we need to weigh a performance bound regarding its tightness and its bounding function, e.g., we need to avoid poor claims like ``the probability that the delay is larger than $30$ seconds is less than $90\%$," which is not helpful in practice.  

Handling the above problems has been a very tricky and intimidating task. Without a clear guideline, it may not be easy to use stochastic network calculus in real-world problems. We are thus motivated in this paper to analyze the above tradeoffs and provide clear guidance on the tricky model selection and model transforms. Although there is a software package, called DISCO~\cite{Sch06}, to ``automatically" derive deterministic performance bounds once model parameters are given, such a software tool has not been seen so far for stochastic network calculus, due to the above tricky tradeoffs. We expect this paper can also clear the road for people who intend to build a software package for stochastic network calculus.

\section{Related Work}

As a new theory for performance evaluation, network calculus has been developed along two tracks: deterministic and stochastic.  Deterministic network calculus~\cite{Chang00,Le} is to search for the worst-case performance bounds, which in many cases are too loose to be useful. Stochastic network calculus~\cite{Ciucu06,Jia,Jiangbook} tries to derive tighter performance bounds, but with a small probability the bounds may not hold true. In practice, the bounds obtained via stochastic network calculus may be more useful, since such bounds present network engineers with a mechanism to utilize statistical multiplexing gain.

It seems a strange phenomenon that most papers on stochastic network calculus mainly focus on the theoretical development. Various types of calculus are proposed to analyze the performance bounds in a stochastic sense~\cite{Ciucu06,Ciucu, Fidler,Jiangbook,Li}. Different approaches have been used, for example, the effective bandwidth~\cite{Ciucu06,Li}, moment generating functions~\cite{Fidler}, Martingale inequality~\cite{Jiang2}. In~\cite{Jia}, a stochastic network calculus is built with quite generic abstract traffic models and service models.  Each calculus, without an exception, has been demonstrated to be effective and useful for some given application scenarios. The quest on using stochastic network calculus to solve queueing problems has been remaining active~\cite{Ciucu,Jiang2}. In contrast, the applications of stochastic network calculus were left behind. It is abnormal that the papers on theoretical development outnumber the ones on realistic applications of this theory.       

The call for a guidance on building suitable stochastic traffic and service models that are simple to obtain and easy to calculate remains unanswered. While many research efforts are being devoted to obtaining tight stochastic bounds close to the exact solutions for special cases~\cite{Ciucu}, we in this paper divert to pursuing the simplicity of the model building methodology.    

\section{Background of Stochastic Network Calculus} \label{sec:background}
\subsection{Notation}
We first introduce the notation and key concepts of stochastic network calculus~\cite{Jia,Jiangbook,Li}. Throughout this paper, we assume that all arrival curves and service curves are non-negative and wide-sense increasing functions. Conventionally, $A(t)$ and $A^*(t)$ are used to denote the \textit{cumulative} traffic that arrives and departures in time interval $(0,t]$, respectively, and $S(t)$ is used to denote the cumulative amount of service provided by the system in time interval $(0,t]$. For any $0\le s \le t$, let $A(s,t) \equiv A(t)-A(s), A^*(s,t) \equiv A^*(t)-A^*(s),$ and $S(s,t) \equiv S(t)-S(s).$ By default, $A(0)=A^*(0)=S(0)=0$. 

We denote by $\mathcal{F}$ the set of non-negative wide-sense increasing functions, i.e.,  $$\mathcal{F}= \{f(\cdot): \forall 0\le x\le y, 0\le f(x) \le f(y)\},$$  and by $\bar{\mathcal{F}}$ the set of non-negative wide-sense decreasing functions, i.e.,  $$\bar{\mathcal{F}}= \{f(\cdot): \forall 0\le x\le y, 0\le f(y) \le f(x)\}.$$ 

For any random variable $X$, its distribution function, denoted by $$F_X (x) \equiv Prob\{X\le x\},$$ belongs to $\mathcal{F}$, and its complementary distribution function, denoted by $$\bar{F}_X(x)\equiv Prob\{X>x\},$$ belongs to $\bar{\mathcal{F}}$. 

During model transform, we may put a stronger requirement on the bounding function. We denote by $\bar{\mathcal{G}}$ the set of functions in $\bar{\mathcal{F}}$ where for each function $g(\cdot) \in \bar{\mathcal{G}}$, its $n$th-fold integration is bounded for any $x \ge 0$ and still belongs to $\bar{\mathcal{G}}$ for any $n \ge 0$, i.e., 
$$ \bar{\mathcal{G}} = \{g(\cdot): \forall n\ge0, \left(\int_x^\infty dy\right)^n g(y) \in \bar{\mathcal{G}}\}.$$

\subsection{Operators}
The following operations defined under the $(\min,+)$ algebra~\cite{Chang00, Cruz91a,Le} will be used in this paper: 
\begin{itemize}
\item The $(\min,+)$ \textit{convolution} of functions $f$ and $g$ is 
\begin{equation}
(f\otimes g) (t) \equiv \inf_{0\le s \le t}\{f(s) + g(t-s)\}.
\end{equation}
 \item The $(\min,+)$ \textit{deconvolution} of functions $f$ and $g$ is 
\begin{equation}
(f\oslash g) (t) \equiv \sup_{s\ge0}\{f(t+s) - g(s)\}.
\end{equation}
\nop{\item The $(\min,+)$ \textit{inf-sum} of functions $f$ and $g$ is
\begin{equation}
(f\odot g) (t) = \inf_{s\ge0}\{f(t+s) + g(s)\}.
\end{equation}}
\end{itemize}
\nop{Note that the $\odot$ operator has not been used before. It is defined in this paper to simplify notation. In addition, we need the normal convolution in our analysis:
\begin{itemize}
\item The \textit{normal convolution} of functions $f$ and $g$ is 
\begin{equation}
(f * g) (x) = \int_0^x f(x-y)dg(y).
\end{equation}
\end{itemize}}
In addition, we adopt:
\begin{itemize}
\item $[x]^+ \equiv max\{x,0\}$,
\item $[x]_1 \equiv min\{x, 1\}$.
\end{itemize}

\subsection{Performance Measures, Traffic and Server Models}
The following measures are of interest in service guarantee analysis under network calculus:
\begin{itemize}
\item The backlog $B(t)$ in the system at time $t$ is defined as:
\begin{equation}
B(t) = A(t) - A^*(t).
\end{equation}
\item The delay $D(t)$ at time $t$ is defined as:
\begin{equation} 
D(t)= \inf \{\tau \ge 0: A(t) \le A^*(t+\tau)\}. \label{delay}
\end{equation}
\end{itemize} 

Stochastic traffic arrival curve and stochastic service curve are core concepts in stochastic network calculus with the former for traffic modeling and the latter for server modeling. It is worth noting that the deterministic traffic arrival curve and the deterministic service curve under the (deterministic) network calculus are a special case of their corresponding stochastic definition. In the literature, there are different definitions of stochastic arrival curve and stochastic service curve \cite{Jia,Jiangbook}. For traffic arrival models, we have:
\begin{definition}\label{def-tac}\textbf{The $t.a.c.$ model: }
A flow $A(t)$ is said to have a \textit{\underline{t}raffic-\underline{a}mount-\underline{c}entric (t.a.c.)} stochastic arrival curve $\alpha \in \mathcal{F}$ with bounding function $f\in \bar{\mathcal{F}}$, denoted by $$A\sim_{ta}<f,\alpha>,$$ if for all $t \ge s \ge 0$ and all $x\ge 0$, it holds~\cite{Jia,Jiangbook} 
\begin{equation}
Prob\{A(s,t)-\alpha(t-s) > x\} \le f(x).
\end{equation}
\end{definition}

\begin{definition}\label{def-vbc} \textbf{The $v.b.c.$ model: }
A flow $A(t)$ is said to have a \textit{\underline{v}irtual-\underline{b}acklog-\underline{c}entric (v.b.c.)} stochastic arrival curve $\alpha \in \mathcal{F}$ with bounding function $f\in \bar{\mathcal{F}}$, denoted by $$A\sim_{vb}<f,\alpha>,$$ if for all $t \ge 0$ and all $x\ge 0$, it holds~\cite{Jia,Jiangbook} 
\begin{equation}
Prob\{\sup_{0\le s \le t} [A(s,t)-\alpha(t-s)] > x\} \le f(x).
\end{equation}
\end{definition}

\begin{definition}\label{def-mbc}\textbf{The $m.b.c.$ model: }
A flow $A(t)$ is said to have a \textit{\underline{m}aximum-virtual-\underline{b}acklog-\underline{c}entric (m.b.c.)} stochastic arrival curve $\alpha \in \mathcal{F}$ with bounding function $f\in \bar{\mathcal{F}}$, denoted by $$A\sim_{mb}<f,\alpha>,$$ if for all $t \ge 0$ and all $x\ge 0$, it holds~\cite{Jia,Jiangbook} 
\begin{equation}
Prob\{\sup_{0\le s \le t} \sup_{0\le u \le s} [A(u,s)-\alpha(s-u)] > x\} \le f(x).
\end{equation}
\end{definition}

For service models, we have the followings.

\begin{definition}\label{def-ws} \textbf{The $w.s.$ model: }
A server is said to provide a flow $A(t)$ with a \textit{\underline{w}eak \underline{s}tochastic (w.s.) service curve} $\beta \in \mathcal{F}$ with bounding function $g \in \bar{\mathcal{F}}$, denoted by $$S\sim_{ws}<g,\beta>,$$ if for all $t\ge 0$ and all $x\ge 0$, it holds~\cite{Jia,Jiangbook}  
\begin{equation}
Prob\{A\otimes\beta(t)-A^*(t)] >x\} \le g(x). 
\end{equation}
\end{definition} 

\begin{definition}\label{def-sc} \textbf{The $s.c.$ model: }
A server is said to provide a flow $A(t)$ with a \textit{\underline{s}tochastic  service \underline{c}urve (s.c.)} $\beta \in \mathcal{F}$ with bounding function $g \in \bar{\mathcal{F}}$, denoted by $$S\sim_{sc}<g,\beta>,$$ if for all $t\ge 0$ and all $x\ge 0$, it holds~\cite{Jia,Jiangbook}  
\begin{equation}
Prob\{\sup_{0\le s \le t}[A\otimes\beta(s)-A^*(s)] >x\} \le g(x). 
\end{equation}
\end{definition} 

\begin{definition}\label{def-ssc} \textbf{The $s.s.c.$ model: }
A server is said to to provide a \textit{\underline{s}trict \underline{s}tochastic service \underline{c}urve (s.s.c.)} $\beta \in \mathcal{F}$ with bounding function $g \in \bar{\mathcal{F}}$, denoted by $$S\sim_{ssc}<g,\beta>,$$ if during any period $(s, t]$ the amount of service $S(s,t)$ provided by the server satisfies~\cite{Jia,Jiangbook}  
\begin{equation}
Prob\{S(s,t) < \beta(t-s) -x \}  \le g(x). 
\end{equation}
\end{definition} 

With the above definitions, various properties of stochastic network calculus, including the stochastic backlog bound and the stochastic delay bound, have been proved (e.g., see~\cite{Jia,Jiangbook,Li}).

Several natural questions arise: Why should we need different forms of traffic arrival models and service models? Can a traffic (service) model be transformed to another traffic (service) model? What is the impact of model transform on performance analysis? 

The first two questions have been answered in~\cite{Jia,Jiangbook}. Briefly speaking, some models are too weak to be useful in the performance bound analysis. For instance, from the $t.a.c.$ model it is hard to obtain the stochastic backlog bound, because according to Lindley equation~\cite{Kle}, $B(t) = \sup_{0\le s\le t} \{A(s,t) - S(s,t)\}$, which requires the calculation of $sup$. The value of $sup$ is not readily obtainable from the $t.a.c.$ model. As such, we may put more constraints on the traffic model, such as those in the $v.b.c.$ model and the $m.b.c.$ model. Regarding the second question, it has been shown that different models can be transformed to each other with the theorems introduced in~\cite{Jiangbook}. The last question, however, has not been touched and is the main focus of the rest of the paper.  

\section{Tradeoffs in Model Transform} \label{sec:Tradeoff}

In this section, we illustrate the tradeoffs in the selection of a proper model. We start from the $t.a.c$ traffic arrival model and the $s.s.c.$ service model, because both of them have the most intuitive meaning and can be obtained easily from the distribution of packet inter-arrival times of the input flow and the distribution of packet service times of the server, respectively.  

We shall ignore the transform from $m.b.c \rightarrow v.b.c. \rightarrow t.a.c.$ and the transform from $s.c \rightarrow w.s.$, because a stronger model (i.e., a model with more constraints) implies a weaker model~\cite{Jiangbook}.  

\subsection{The Transform from $t.a.c.$ to $v.b.c.$}
Although the $t.a.c.$ model is the most simple model, it is not easy to derive performance bounds with this model~\cite{Jiangbook}. As such, we need to transform it to a stronger model, e.g., the $v.b.c.$ model or the $m.b.c.$ model.

\begin{remark}
It is very likely that after the transform the values of the bounding function go to a large probability value or even $1$ when time $t$ goes to $\infty$, especially when the bounding function of the $t.a.c.$ curve is dependent on time $t$. This is the so-called time-increasing problem on the bounding function. To avoid this problem, it is suggested~\cite{Li} that there should exist a time scale $T$ enforced on the traffic and the service. We follow same idea assume that all traffic arrival curves and service curves are enforced on the maximum time scale $T$, e.g., the constraint of $v.b.c.$ traffic curve becomes $$Prob\{\sup_{t-T\le s \le t} [A(s,t)-\alpha(t-s)] > x\} \le f(x).$$ Later, we will discuss other methods to tackle this problem. 
\end{remark}

\begin{lemma} ~\label{Th1:tacTovbc}
If a stationary traffic flow has a $t.a.c$ stochastic arrival curve $\alpha \in \mathcal{F}$ with bounding function $f \in \bar{\mathcal{G}}$, it also has a $v.b.c.$ stochastic arrival curve $\alpha_\theta \in \mathcal{F}$ with bounding function $f^\theta \in \bar{\mathcal{G}}$, where for any $\theta >0$ 
\begin{eqnarray}
&&\alpha_\theta= \alpha(t) + \theta \cdot t,  \label{eq:1} \\
&&f^\theta(x) = \left[\frac{1}{\theta}\int_{x}^{x+T\theta} f(y) dy\right]_1. \label{eq:2} 
\end{eqnarray} 
\end{lemma} 

\begin{proof}
Since the traffic flow is stationary, for any $\theta >0, t_1 \ge T, t \le T$, we have 
\begin{align}
 & \sup_{t_1-t \le s \le t_1} \{A(s,t) - \alpha_\theta (t-s)\} \nonumber \\
 &= \sup_{0 \le s \le t} \{A(s,t) - \alpha_\theta (t-s)\} \nonumber  \\
 & \le  \sup_{0 \le s \le T} \{A(s,t) - \alpha_\theta (t-s)\} \nonumber \\
    & \le  \sup_{0 \le s \le T} \{A(s,t) - \alpha_\theta (t-s)\}^+. \label{eq:proof1}
\end{align} 
Since for any $x\ge 0$, $Prob\{[A(s,t) - \alpha_\theta (t-s)]^+ >x\} = Prob\{A(s,t) - \alpha_\theta (t-s) >x\} \le f(x+\theta\cdot(t-s))$, we have 
\begin{align}
 & Prob\{\sup_{t_1-t \le s \le t_1} \{A(s,t) - \alpha_\theta (t-s)\} >x\} \nonumber \\
 & \le  Prob\{\sup_{0 \le s \le T} \{A(s,T) - \alpha_\theta (T-s)\} > x\} \nonumber \\
 & \le  \sum_{s=0}^{T} Prob \{[A(s,T) - \alpha_\theta(T-s)]^+ >x \} \nonumber \\
 & \le  \sum_{s=0}^{T} f(x+\theta\cdot(T-s)) \nonumber \\
 & \le \frac{1}{\theta}\int_{x}^{x+T\theta} f(y) dy. 
\end{align} 
The theorem holds since the probability has to be smaller than $1$. 
\end{proof}

Essentially, Lemma~\ref{Th1:tacTovbc} indicates that for a  flow following a $t.a.c.$ traffic arrival curve, we can model the same flow with a series of $v.b.c.$ traffic arrival curves. An interesting question is that among these $v.b.c.$ traffic arrival curves, which one is the best with which we can obtain the tightest performance bounds? To evaluate, we need to formally define the tightness of traffic arrival curves in the stochastic sense.

\begin{definition} ~\label{Def1} \textbf{(Stochastic tightness of traffic arrival curves)} Assume that a traffic flow $A(t)$ follows a $t.a.c.$ (or $v.b.c.$, $m.b.c.$) traffic arrival curve $\alpha_1$ with bounding function $f_1$ as well as a  $t.a.c.$ (or $v.b.c.$, $m.b.c.$, respectively) traffic arrival curve $\alpha_2$ with bounding function $f_2$. We call the curve $\alpha_1$ is stochastically tighter than the curve $\alpha_2$ within a tolerance bound $\epsilon \ge 0$, denoted by $\alpha_1 <_\epsilon \alpha_2$, if for all $t \ge 0 $ and all $x\ge 0$, there hold
\begin{equation} \label{eq:3}
\alpha_1(t)  \left\{ \begin{array}{rl}
 \le \alpha_2(t) &\mbox{ if $t=0$} \\
  < \alpha_2(t) &\mbox{ otherwise,}
  \end{array} \right.
\end{equation}
and 
\begin{eqnarray}
f_1(x) \le f_2(x) + \epsilon. \label{eq:4}
\end{eqnarray} 
If $\epsilon =0$, we also say that $\alpha_1$ is absolutely tighter than $\alpha_2$. 
\end{definition}

Generally speaking, we need to make a tradeoff between the arrival curve and its bounding function. From Lemma~\ref{Th1:tacTovbc}, we observe that the series of $v.b.c.$ traffic arrival curves depends on the value of $\theta$. We should not select a very loose traffic arrival curve (i.e., a very large $\theta$ value) to make the bounding function small; nor should we use a very tight traffic arrival curve (i.e., a very small $\theta$ value) such that the bounding function becomes not useful. For example, ``the probability that a certain event occurs is no larger than $1$" is meaningless.     

We have the following theorem to determine another tighter traffic arrival curve based on an existing traffic arrival curve and a given acceptable range on the bounding function. 

\begin{lemma}\label{Th2} 
Assume that a traffic flow has a $t.a.c$ stochastic arrival curve $\alpha \in \mathcal{F}$ with the bounding function $f \in \bar{\mathcal{G}}$. Assume that one of its corresponding $v.b.c.$ stochastic arrival curves is $\alpha_{\theta_1} \in \mathcal{F}$ with the bounding function $f^{\theta_1} \in \bar{\mathcal{G}}$. We can model it with another $v.b.c.$ stochastic arrival curve $\alpha_{\theta_2}\in \mathcal{F}$ with the bounding function $f^{\theta_2}\in \bar{\mathcal{G}}$ such that   $\alpha_{\theta_2} <_\epsilon \alpha_{\theta_1}$ if there exists $\theta_2 < \theta_1$ for any $x>0$ satisfying
\begin{equation} \label{Th2:condition}
\int_{x}^{x+T\theta_2}\frac{1}{\theta_2}f(y)dy- \int_{x}^{x+T\theta_1}\frac{1}{\theta_1}f(y)dy \le \epsilon.
\end{equation}
\end{lemma}  
\nop{\begin{proof}
For the bounding function to be useful, we assume that after the transform, the bounding function is strictly less than $1$. Based on Lemma~\ref{Th1:tacTovbc}, we may obtain two $v.b.c.$ models, denoted by $<f^{\theta_1}, \alpha_{\theta_1}>$ and $<f^{\theta_2}, \alpha_{\theta_2}>$, respectively, where, 
\begin{eqnarray}
&&\alpha_{\theta_1}= \alpha(t) + \theta_1 \cdot t,  \label{eq:1.1} \\
&&f^{\theta_1}(x) =  \frac{1}{\theta_1}\int_{x}^{x+T\theta_1} f(y) dy \label{eq:2.1} \\
&&\alpha_{\theta_2}= \alpha(t) + \theta_2 \cdot t,  \label{eq:1.2} \\
&&f^{\theta_2}(x) = \frac{1}{\theta_2}\int_{x}^{x+T\theta_2} f(y) dy \label{eq:2.2}.  
\end{eqnarray} 
If $\alpha_{\theta_2} <_\epsilon \alpha_{\theta_1}$, according to Definition~\ref{Def1} we must have $\alpha_{\theta_2} < \alpha_{\theta_1}$ and $f_2(x) \le f_1(x) + \epsilon$. From (\ref{eq:1.1}) and (\ref{eq:1.2}), we have $\theta_2 <\theta_1,$ and from (\ref{eq:2.1}) and (\ref{eq:2.2}), we have $\frac{\theta_1 \int_x^\infty f(y) dy} { \int_x^\infty f(y) dy + \theta_1 \epsilon} \le \theta_2.$
\end{proof}
}

Lemma~\ref{Th2} is easy to prove based on Lemma~\ref{Th1:tacTovbc} and Definition~\ref{Def1}.

\begin{remark}
Based on Lemma~\ref{Th2}, if no such $\theta_2$ could be found, we call the arrival curve $\alpha_{\theta_1}$ the stochastically tightest within the tolerance bound $\epsilon$. Given any $v.b.c.$ traffic curve, Lemma~\ref{Th2} is useful in searching for a tighter stochastic arrival curve, if exists. For example, assume that we have a $v.b.c$ traffic curve $<f^{\theta_1}, \alpha_{\theta_1}>$, where $\theta_1$ is known. Denote 
$$\phi(\theta_2, x)= \int_{x}^{x+T\theta_2}\frac{1}{\theta_2}f(y)dy- \int_{x}^{x+T\theta_1}\frac{1}{\theta_1}f(y)dy.$$ 
Setting a lower threshold value on x, say $\underline{x}$, and a tolerance bound $\epsilon$, we can check if the equation $\phi(\theta_2, \underline{x}) = \epsilon$ has a positive root on $\theta_2$. If no solution could be found, $<f^{\theta_1}, \alpha_{\theta_1}>$ is the tightest $v.b.c.$ curve within the tolerance bound $\epsilon$. 
\end{remark}

\subsection{The Transform from $v.b.c.$ to $m.b.c.$}

\begin{lemma} ~\label{Th3:vbcTombc}
If a traffic flow has a $v.b.c$ stochastic arrival curve $\alpha \in \mathcal{F}$ with bounding function $f \in \bar{\mathcal{G}}$, it also has a $m.b.c.$ stochastic arrival curve $\alpha_\theta$ with bounding function $f^\theta \in \bar{\mathcal{G}}$, where for any $\theta >0$ 
\begin{eqnarray}
&&\alpha_\theta= \alpha(t) + \theta \cdot t,  \label{eq:3.1} \\
&&f^\theta(x) = \left[\frac{1}{\theta}\int_{x-\theta T}^{x} f(y) dy\right]_1 \label{eq:3.2}.  
\end{eqnarray} 
\end{lemma} 

The proof of Lemma~\ref{Th3:vbcTombc} is similar to that of Lemma~\ref{Th1:tacTovbc} and is omitted. Due to the similarity between Lemma~\ref{Th3:vbcTombc} and Lemma~\ref{Th1:tacTovbc}, we can slightly revise Lemma~\ref{Th2} so that we can check whether there exists a stochastically tighter $m.b.c.$ curve, given an existing $m.b.c.$ curve.  

\subsection{The Transforms from $s.s.c.$ to $w.s$ and $s.c.$}

\begin{lemma}~\label{Th4}
Consider a server that provides a stochastic strict service curve $\beta \in \mathcal{F}$ with bounding function $g(x) \in \bar{\mathcal{F}}$. 
\begin{enumerate}
\item The server also provides a weak stochastic service curve $\beta(t)$ with the same bounding function $g(x)$. 
\item If $g(x) \in \bar{\mathcal{G}}$ and the input and output processes are both stationary, the server provides a stochastic service curve $\beta_{-\theta}$ with bounding function $g^\theta(x)$, where for any $\theta >0$,
 \begin{eqnarray}
&&\beta_{-\theta}= \beta(t) - \theta \cdot t,\label{eq:4.1} \\
&&g^\theta(x) = \left[\frac{1}{\theta} \int_{x-\theta\cdot T +\theta}^{x} f(y) dy\right]_1. \label{eq:4.2} 
\end{eqnarray}
\end{enumerate}
\end{lemma}

\begin{proof}Please refer to chapter 4 of~\cite{Jiangbook} for the proof of the first part. 

For the second part, we first have for any $t\ge s$,
$$A\otimes \beta(s) \le A\otimes \beta(s) - \theta (t-s), $$
and hence for any $t_1 \ge T, t \le T$, 
\begin{align}
& Prob\{\sup_{t_1-T \le s \le t_1} [A\otimes \beta_{-\theta}(s) - A^*(s)] >x\} \nonumber \\
\le &  Prob\{\sup_{0\le s \le T} [A\otimes \beta_{-\theta}(s) - A^*(s)] >x\} \nonumber \\
\le & Prob\{\sup_{1 \le s \le T} [A\otimes \beta(s) - A^*(s)-\theta\cdot s]^+ >x-\theta\cdot T\} \nonumber\\
\le & \sum_{s=1}^T Prob\{[A\otimes \beta(s) - A^*(s)-\theta\cdot s]^+ >x-\theta\cdot T\} \nonumber \\
\le & \sum_{s=1}^T f(x-\theta\cdot T +\theta\cdot s) \label{eq:4proof} \\
\le & \frac{1}{\theta} \int_{x-\theta\cdot T +\theta}^{x} f(y) dy 
\end{align}
Since the probability cannot be larger than 1, the second part is proved. Note that the inequality~(\ref{eq:4proof}) is true due to the first part of the theorem. 
\end{proof}

Lemma~\ref{Th4} indicates that if we transform the $s.s.c$ service model to the $s.c.$ service model, we obtain a series of $s.c.$ curves and are faced with the problem of selecting a ``good" $s.c.$ curve for performance analysis. Similar to the transform of traffic models, we need to define the tightness of service curves.  

\begin{definition} ~\label{Def2} \textbf{(Stochastic tightness of service curves)} Assume that a service provided by a system follows an $s.s.c.$ (or $w.s.$, $s.c.$)  service curve $\beta_1$ with bounding function $g_1$ as well as an $s.s.c.$ (or $w.s.$, $s.c.$, correspondingly) service curve $\beta_2$ with bounding function $g_2$. We call the curve $\beta_1$ is stochastically tighter than the curve $\beta_2$ within a tolerance bound $\epsilon \ge 0$, denoted by $\beta_1 >_\epsilon \beta_2$, if for all $t \ge 0 $ and all $x\ge 0$, there hold
\begin{equation} \label{eq:def2-1}
\beta_1(t)  \left\{ \begin{array}{rl}
 \ge \beta_2(t) &\mbox{ if $t=0$} \\
  > \beta_2(t) &\mbox{ otherwise,}
  \end{array} \right.
\end{equation}
and 
\begin{eqnarray}
g_1(x) \le g_2(x) + \epsilon. \label{eq:def2-2}
\end{eqnarray} 
If $\epsilon =0$, we also call that $\beta_1$ is absolutely tighter than $\beta_2$. 
\end{definition}

When we select a good stochastic service curve, we again need to make the tradeoff between the tightness of the service curve and the usefulness of the bounding function. Similar to the previous section, we have the following theorem to help select a good stochastic service curve after the model transform.  

\begin{lemma}\label{Th5} 
Assume that a server provides an $s.s.c$ service curve $\beta \in \mathcal{F}$ with the bounding function $g \in \bar{\mathcal{G}}$. Assume that one of its corresponding $s.c.$ service curves is $\beta_{-\theta_1} \in \mathcal{F}$ with the bounding function $g^{\theta_1} \in \bar{\mathcal{F}}$. We can model it with another $s.c.$ service curve $\beta_{-\theta_2}\in \mathcal{F}$ with the bounding function $g^{\theta_2}\in \bar{\mathcal{F}}$ such that   $\beta_{-\theta_2} >_\epsilon \beta_{-\theta_1}$ if there exists $\theta_2 < \theta_1$ for any $x>0$ satisfying
\begin{equation} \label{Th5:condition}
\int_{x-\theta_2+\theta_2}^x \frac{1}{\theta_2}f(y)dy - \int_{x-\theta_1+\theta_1}^x \frac{1}{\theta_1}f(y)dy \le \epsilon.
\end{equation}
\end{lemma}  

Lemma~\ref{Th5} is easy to prove based on Lemma~\ref{Th4} and Definition~\ref{Def2}. 

\nop{Another difficulty is to handle the $sup$ operation in
Definitions~\ref{def-sac} and~\ref{def-sc}. Methods have been provided
in~\cite{Jiangbook} to transform one type of arrival curve (service
curve) to another, e.g., (\ref{eq:PoissonApproximation}) can be
transformed to an equivalent m.b.c. stochastic arrival curve with extra constraints.
Alternatively, we can use \[Prob\{\sup_{s\ge 0} X_s >x\} \approx
\max_{s\ge 0}\{Prob\{X_s >x\}\}\] to simplify the calculation. The
same approximation has been used in~\cite{Kni} and is adopted in
the numerical calculation of this paper.
}

\section{The Impact of Model Transform on Performance Evaluation} 
In this section, we illustrate the impact of using different traffic arrival/ service curves on the performance evaluation.  We only use the output characteristics and the service guarantee (e.g., stochastic bounds on delay and backlog) as the examples. The impact on other properties such as the leftover services is omitted to save space. 

\begin{lemma} \label{lemma:1}
\begin{enumerate}
\item If for any $x \ge 0$, $f_1(x) \le f_2(x)+ \epsilon_1$ and $g_1(x) \le g_2(x)+ \epsilon_2$,  then $f_1\otimes g_1(x) \le f_1\otimes g_2(x) + \epsilon_1+\epsilon_2.$  
\item If for any $x \ge 0$, $\alpha_1(x) \le \alpha_2(x)$ and $\beta_1(x) \ge \beta_2(x)$, then $\alpha_1 \oslash \beta_1(x) \le \alpha_2 \oslash \beta_2(x).$
\item If for any $x \ge 0$, $\alpha_1(x) \le \alpha_2(x)$ and $\beta_1(x) \ge \beta_2(x)$, then $h(\alpha_1, \beta_1) \le h(\alpha_2, \beta_2),$ where $h(\alpha, \beta)$ represents the maximum horizontal distance between functions $\alpha$ and $\beta$, i.e., 
$$h(\alpha, \beta) = \sup_{s\ge0}\{\inf \{\tau \ge0: \alpha(s) \le \beta(s+\tau\}\}.$$ 
\end{enumerate}
\end{lemma}
\begin{proof}
For the first part, we have:
\begin{align}
f_1 \otimes g_1(x) & = \inf_{0\le y \le x} \{f_1(y) + g_1(x-y)\} \nonumber \\
&\le \inf_{0\le y \le x} \{f_2(y)+ \epsilon_1 + g_2(x-y) + \epsilon_2 \} \nonumber \\
& =  \inf_{0\le y \le x} \{f_2(y)+ g_2(x-y)\} + \epsilon_1+\epsilon_2 \nonumber \\
& = f_2 \otimes g_2 (x) + \epsilon_1+\epsilon_2. \nonumber  
\end{align}

For the second part, we have: 
\begin{align}
\alpha_1 \oslash \beta_1 (x) & = \sup_{y \ge 0} \{\alpha_1(x+y) - \beta_1(y)\} \nonumber \\
& \le \sup_{y \ge 0} \{\alpha_2(x+y) - \beta_2(y)\} \nonumber \\
& = \alpha_2 \oslash \beta_2 (x). \nonumber
\end{align}

To prove the last part, we have: 
\begin{align}
h(\alpha_1, \beta_1) &= \sup_{s\ge0}\{\inf \{\tau \ge0: \alpha_1(s) \le \beta_1(s+\tau\}\} \nonumber \\
& \le\sup_{s\ge0}\{\inf \{\tau \ge0: \alpha_2(s) \le \beta_1(s+\tau\}\} \nonumber \\
& \le \sup_{s\ge0}\{\inf \{\tau \ge0: \alpha_2(s) \le \beta_2(s+\tau\}\} \nonumber \\
& = h(\alpha_2, \beta_2). \nonumber 
\end{align}
\end{proof}

From Lemma~\ref{lemma:1} and Theorem 5.12 in~\cite{Jiangbook}, it is easy to have the following theorem:
\begin{lemma} \label{Th6}
Consider a system with input $A$. Assume that  $A\sim_{vb} <f_1, \alpha_1>$ as well as $A\sim_{vb}<f_2, \alpha_2>$, and the system provides the input $A$ with a service that can be modeled by $S\sim_{sc} <g^1, \beta_1>$ and $S\sim_{sc} <g^2, \beta_2>$. If $\alpha_1$ is stochastically tighter than $\alpha_2$ within the tolerance bound $\epsilon_1$ and $\beta_1$ is stochastically tighter than $\beta_2$ within the tolerance bound $\epsilon_2$, then the output $A^*$ can be modeled by:  
\begin{enumerate}
\item $A^*\sim_{vb} <f_1\otimes g_1, \alpha_1 \oslash \beta_1>,$ or 
\item $A^*\sim_{vb} <f_1\otimes g_2, \alpha_1 \oslash \beta_2>,$ or
\item $A^*\sim_{vb} <f_2\otimes g_1, \alpha_2 \oslash \beta_1>,$ or
\item $A^*\sim_{vb} <f_2\otimes g_2, \alpha_1 \oslash \beta_2>,$ 
\end{enumerate}
among which $A^*\sim_{vb} <f_1\otimes g_1, \alpha_1 \oslash \beta_1>$ is the stochastically tightest within tolerance bound $\epsilon_1+\epsilon_2$. 
\end{lemma}

From Lemma~\ref{lemma:1} and Theorem 5.4 in~\cite{Jiangbook}, it is easy to have:  

\begin{lemma} \label{Th7}
Consider a system with input $A$. Assume that  $A\sim_{vb} <f_1, \alpha_1>$ as well as $A\sim_{vb}<f_2, \alpha_2>$, and the system provides the input $A$ with a service that can be modeled by $S\sim_{ws} <g^1, \beta_1>$ and $S\sim_{ws} <g^2, \beta_2>$. If $\alpha_1$ is stochastically tighter than $\alpha_2$ within the tolerance bound $\epsilon_1$ and $\beta_1$ is stochastically tighter than $\beta_2$ within the tolerance bound $\epsilon_2$, then for all $t \ge 0$ and $x\ge 0$ the delay $D(t)$ can be bounded by : 
\begin{enumerate}
\item $Prob\{D(t) > h(\alpha_1+x, \beta_1) \}\le f^1\otimes g^1(x),$ or
\item $Prob\{D(t) > h(\alpha_1+x, \beta_2) \}\le f^1\otimes g^2(x),$ or
\item $Prob\{D(t) > h(\alpha_2+x, \beta_1) \}\le f^2\otimes g^1(x),$ or
\item $Prob\{D(t) > h(\alpha_2+x, \beta_2) \}\le f^2\otimes g^2(x),$
\end{enumerate}
among which $Prob\{D(t) > h(\alpha_1+x, \beta_1) \}\le f^1\otimes g^1(x)$ is the stochastically tightest bound in the sense that  $h(\alpha_1+x, \beta_1)$ is the smallest for all $h(\alpha_i+x, \beta_j), i, j =1, 2$, and $f^1\otimes g^1 -\epsilon_1-\epsilon_2$ is also the smallest for all $f^i\otimes g^j, i, j=1,2.$ 
\end{lemma}

Regarding the backlog bound, we have: 
\begin{lemma} \label{Th8}
Consider a system with input $A$. Assume that  $A\sim_{vb} <f_1, \alpha_1>$ as well as $A\sim_{vb}<f_2, \alpha_2>$, and the system provides the input $A$ with a service that can be modeled by $S\sim_{ws} <g^1, \beta_1>$ and $S\sim_{ws} <g^2, \beta_2>$. If $\alpha_1$ is stochastically tighter than $\alpha_2$ within the tolerance bound $\epsilon_1$ and $\beta_1$ is stochastically tighter than $\beta_2$ within the tolerance bound $\epsilon_2$, then for all $t \ge 0$ and $x\ge 0$ the backlog $B(t)$ can be bounded by : 
\begin{enumerate}
\item $Prob\{B(t) > x \}\le f^1\otimes g^1(x- \alpha_1\oslash \beta_1(0)),$ or
\item $Prob\{B(t) > x \}\le f^1\otimes g^2(x- \alpha_1\oslash \beta_2(0)),$ or
\item $Prob\{B(t) > x \}\le f^2\otimes g^1(x- \alpha_2\oslash \beta_1(0)),$ or
\item $Prob\{B(t) > x \}\le f^2\otimes g^2(x- \alpha_2\oslash \beta_2(0)),$
\end{enumerate}
among which $Prob\{B(t) > x \}\le f^1\otimes g^1(x- \alpha_1\oslash \beta_1(0))$ is the stochastically tightest bound in the sense that   $f^1\otimes g^1(x- \alpha_1\oslash \beta_1(0))$ is the smallest for all $ f^i\otimes g^j(x- \alpha_i\oslash \beta_j(0))-\epsilon_1-\epsilon_2, i,j=1,2.$ 
\end{lemma}

Lemma~\ref{Th8} can be easily proved based on Theorem 5.1 in~\cite{Jiangbook}, Lemma~\ref{lemma:1}, and the fact that if $f, g \in \bar{\mathcal{F}}$ then $f\otimes g \in \bar{\mathcal{F}}.$

\section{Integrating Queueing Theory and Stochastic Network Calculus }
\nop{\subsection{A Recipe for Easy Performance Analysis}}
 
In~\cite{Ciucu,Jiang2}, performance bounds close to the exact solutions with traditional queueing theory have been derived with stochastic network calculus. These methods derive performance bounds based on concentration inequality and martingale inequality. In addition, it is required to solve optimization problems in the selection of the best parameters to obtain tight bounds~\cite{Ciucu}. These methods are effective but the complexity in model buildup may make them not easily accessible to a broad range of users. 

We suggest that model building and performance analysis should be decoupled to ease the application of stochastic network calculus. More specifically, it should be a two-phase process. In the first phase, tight traffic arrival and service models should be built, and all the complexity requiring the help of traditional queueing analysis should be dealt with in this phase. The second phase is to \textit{automatically} calculate performance bounds with standard operations in stochastic network calculus. It should be feasible to build a software tool to handle the bound analysis in the second phase.    

Traditional queueing theory plays an important role in the first phase. While it is hard to suggest a general approach suitable for all applications, the following guidelines should be helpful. 
\begin{itemize}
\item  \textbf{Rule 5:} Regarding traffic arrival curves, we can build a virtual queueing system, with the traffic arrival process as the input and the arrival curve as the service of the virtual queueing system. Building the $v.b.c$ (or $m.b.c.$) model is equivalent to finding the queue length (or maximum queue length, respectively) of the virtual queueing system with respect to time $t$. This is a traditional queuing theory problem, and tight bound can usually be found.
\item  \textbf{Rule 6:} Regarding service curves, both $w.s.$ and $s.c.$ models are coupled with the input process (in definition). As such, we suggest starting with the $s.s.c.$ model, which has a straightforward meaning. From the $s.s.c.$ model, we can directly get the $w.s.$ model. In the following, we propose another new service model, which is easy to find with queuing theory and can be used to obtain both $w.s.$ and $s.c.$ models directly. 
\end{itemize}   

\begin{definition}  \textbf{(Virtual backlog stochastic strict service curve).} A system is said to be a virtual backlog stochastic strict server providing service curve $\beta(t)$ with bounding function $g(x) \in \bar{\mathcal{F}}$, denoted by $S\sim_{vbssc}<g(x), \beta(t)>$, if during any period $(s,t]$ the amount of service $S(s,t)$ provided by the system satisfies 
\begin{equation}
Prob\{\sup_{0\le s\le t} [\beta(t-s) - S(s,t)] > x \} \le g(x),
\end{equation} 
for any $x\ge 0$.
\end{definition}

We have the following theorem: 
\begin{theorem} \label{Th:vbssc}
Consider a system that is a virtual backlog stochastic strict server providing service curve $\beta(t)$ with bounding function $g(x) \in \bar{\mathcal{F}}$. It also provides an $s.c.$ stochastic service curve $\beta(t)$ with the same bounding function $g(x)$. 
\end{theorem}
\begin{proof}
For any $t \le 0$, $\forall 0\le s\le t$, we have 
\begin{itemize}
\item \textbf{Case 1:}  $s$ is not within any backlogged period. In this case, $A^*(s) = A(s)$. Therefore, 
\begin{equation} \label{eq:proof}
A\otimes \beta(s) - A^*(s) \le A(s) + \beta(0) -A^*(s) \le 0.
\end{equation}
The last inequality is because $\beta(0)$ needs to lower bound the service, which is $S(0)=0$ by default. 
\item \textbf{Case 2:}  $s$ is within a backlogged period. Assume that the backlogged period starts from $s_0 \le s$. Then $A^*(s_0) = A(s_0)$, and 
\begin{align} 
&A\otimes \beta(s) - A^*(s) \nonumber \\
& \le A(s_0) + \beta(s-s_0) - A^*(s) \nonumber \\
& = \beta(s-s_0) + A^*(s_0) - A^*(s) \nonumber \\
& = \beta(s-s_0) - S(s_0, s) \label{eq:proof2}
\end{align}
\end{itemize}
Combining both cases, we have $$A\otimes \beta(s) - A^*(s) \le \beta(s-s_1) - S(s_1,s)$$ 
holds for all $s \in [0, t]$ and $0\le s_1 \le s$. Therefore,   
\begin{align}
&\sup_{0\le s\le t} [A\otimes \beta(s) - A^*(s)] \nonumber \\
& \le \sup_{0 \le s \le t, 0\le s_1\le s} [\beta(s-s_1) - S(s_1,s)] \nonumber \\
& = \sup_{0 \le s \le t} [\beta(t-s) - S(s,t)],
\end{align}
from which and the definitions of $s.c.$ and $v.b.s.s.c.$ service curves the theorem is proved.
\end{proof}

Regarding service curves, we have the following suggestion based on Theorem~\ref{Th:vbssc}: 
\begin{itemize}
\item  \textbf{Rule 7:} For a server providing service process $S(t)$, we can build a virtual queueing system, with the service curve $\beta(t)$ as the input and the service process $S(t)$ as the service of the virtual queueing system. Building the $v.b.s.s.c.$ model is equivalent to finding the queue length of the virtual queueing system with respect to time $t$. Again, this is a traditional queueing theory problem. With the $v.b.s.s.c.$ model, we can immediately obtain the $s.c.$ or $w.s.$ model. 
\end{itemize}
 
\nop{\subsection{Revisiting $M/M/1$} }

\nop{In this section, we demonstrate that Doob's Martingale inequalities~\cite{Feller} are a very powerful tool in building tight stochastic traffic arrival and service models without requiring model transform.}

\section{An Example and Guidelines}

\nop{Given the typical $M/M/1$ queueing model, what are the best bounds that we can obtain with only model transform in stochastic network calculus? We note that by considering the stochastic features of particular traffic arrivals and/or services, better bounds may be available with queueing analysis~\cite{Fry,Kingman}, but those complicate methods are distribution-specific and need to study the intricate interaction between the input and the service. They are not the focus of this paper.  }

We use the typical $M/M/1$ queueing model as an example to illustrate how an intuitive and nature model transform may result in poor performance bounds. Drawn from this example, we summarize the implications of model transform, and provide the insights on the possible problems in the model selection and transform.   

\subsection{An Example}

Assume that a flow has fixed unit packet size. Assume that its packets arrive according to a Poisson arrival process
with mean rate $\lambda$. Assume that the service time of each packet in the server has an exponential distribution with mean $\frac{1}{\mu}$, where $\mu > \lambda$ so that the system is stable. 

This simple $M/M/1$ model has exact solutions on performance bounds with traditional queueing analysis~\cite{Kle} and similar results can be obtained with special treatment on traffic models and service models of stochastic network calculus~\cite{Ciucu}. Nevertheless, we would start from some simple and ``intuitively" right models and are interested in the impact of model transform on performance bounds. Naturally, the Poisson traffic arrival process and the exponential server can be easily modeled by a $t.a.c.$ traffic arrival curve and an $s.s.c.$ service curve. To obtain the performance bounds, we then need to transfer the $t.a.c.$ curve to a $v.b.c.$ curve, according to the theorems in the previous section. 

\textbf{Traffic Model:} It is easy to see that in any time interval $(s, s+t]$, for any $x \ge 0$:
\begin{equation} \label{eq:PoissonTraffic}
Prob\{A(t) -\lambda t > x\} \le  \sum\limits_{k  = \left\lceil {x  +
\lambda t} \right\rceil }^\infty \frac{e^{ - \lambda t}  \cdot
(\lambda t)^{k}}{k!}.
\end{equation}

It is hard to transform the above $t.a.c.$ traffic arrival curve to other traffic models due to the sum in the bounding function. To make the model transform easy, we derive a simpler bounding function based on the Poisson approximation~\cite{Mic}. 

\begin{lemma}\label{Lemma:PA} A Poisson arrival process $A(t)$ with mean rate $\lambda$ is bounded by
\begin{equation}
Prob\{A(t) -\lambda t > x\} \le e^{x - (\lambda t +
x)In\frac{{(\lambda t + x)}}{{\lambda t}}}
\label{eq:PoissonApproximation}.
\end{equation}
\end{lemma}

The proof of Lemma~\ref{Lemma:PA} is in the appendix. Note that when $t$ goes to $\infty$, the right-hand side of (\ref{eq:PoissonApproximation}) goes to $1$, indicating that we have the time-increasing problem with the above bounding function. To avoid the problem, we assume that there is a maximum time scale $T$ enforced on the traffic arrivals and the service. Therefore, for any $s\ge 0$ and any $0\le t \le T$, 
\begin{equation}
Prob\{A(s, s+t) -\lambda t > x\} \le e^{x - (\lambda T +
x)In\frac{{(\lambda T + x)}}{{\lambda T}}}
\label{eq:PoissonApproximation2}.
\end{equation}

With Theorem~\ref{Th1:tacTovbc}, we can find a series of corresponding $v.b.c.$ curves $A(t) \sim_{vb} <\alpha_\theta, f_\theta>$, where for any $0\le t \le T$
\begin{equation} 
\alpha_\theta= (\lambda+\theta)\cdot t,  \label{eq:Poisson1} \\
\end{equation}
\begin{equation} \label{eq:Poisson2}
\begin{split}
f^\theta(x) = & [ e^{x - (\lambda T +
x)In\frac{{(\lambda T + x)}}{{\lambda T}}}
 + \\ 
 &\frac{1}{\theta}\int_x^{x+\theta T} e^{y - (\lambda T +
y)In\frac{{(\lambda T + y)}}{{\lambda T}}} dy ]_1 
\end{split}
\end{equation}

\textbf{Service Model:} During any backlogged period $(s,s+t]$, it is known that the packets departing from the server has an exponentially distributed inter-arrival times with mean $\frac{1}{\mu}$, that is, the departure has a Poisson process with mean $\frac{1}{\mu}$ during the backlogged period. Therefore, 
\begin{equation}
Prob\{S(s,s+t) =n\} = \frac{\mu^n t^n}{n!}e^{-\mu t},
\end{equation}
from which we can get for $\beta(t) = \mu t$, 
\begin{equation}
Prob\{S(s,s+t) < \beta(t) -x\} \le \sum\limits_{k=0}^{\left\lceil \mu t -x \right\rceil } \frac{\mu^k t^k}{k!}e^{ - \mu t}.  
\end{equation}
Using Poisson approximation, the bounding function $g(x)$ becomes: 
\begin{equation}
g(x) = 1- e^{-x - (\mu T - x)In\frac{{(\mu T - x)}}{{\mu T}}}.
\end{equation}

If a server has an $s.s.c$ model, it also has the $s.c.$ mode with the same service curve and bounding function. Having the combination of a $v.b.c$ traffic model and an $s.c.$ model, we can then derive performance bounds.  

\begin{figure*}[htb!]
\begin{center} \begin{minipage}{75 mm}
    \centerline{\includegraphics[width=60mm]{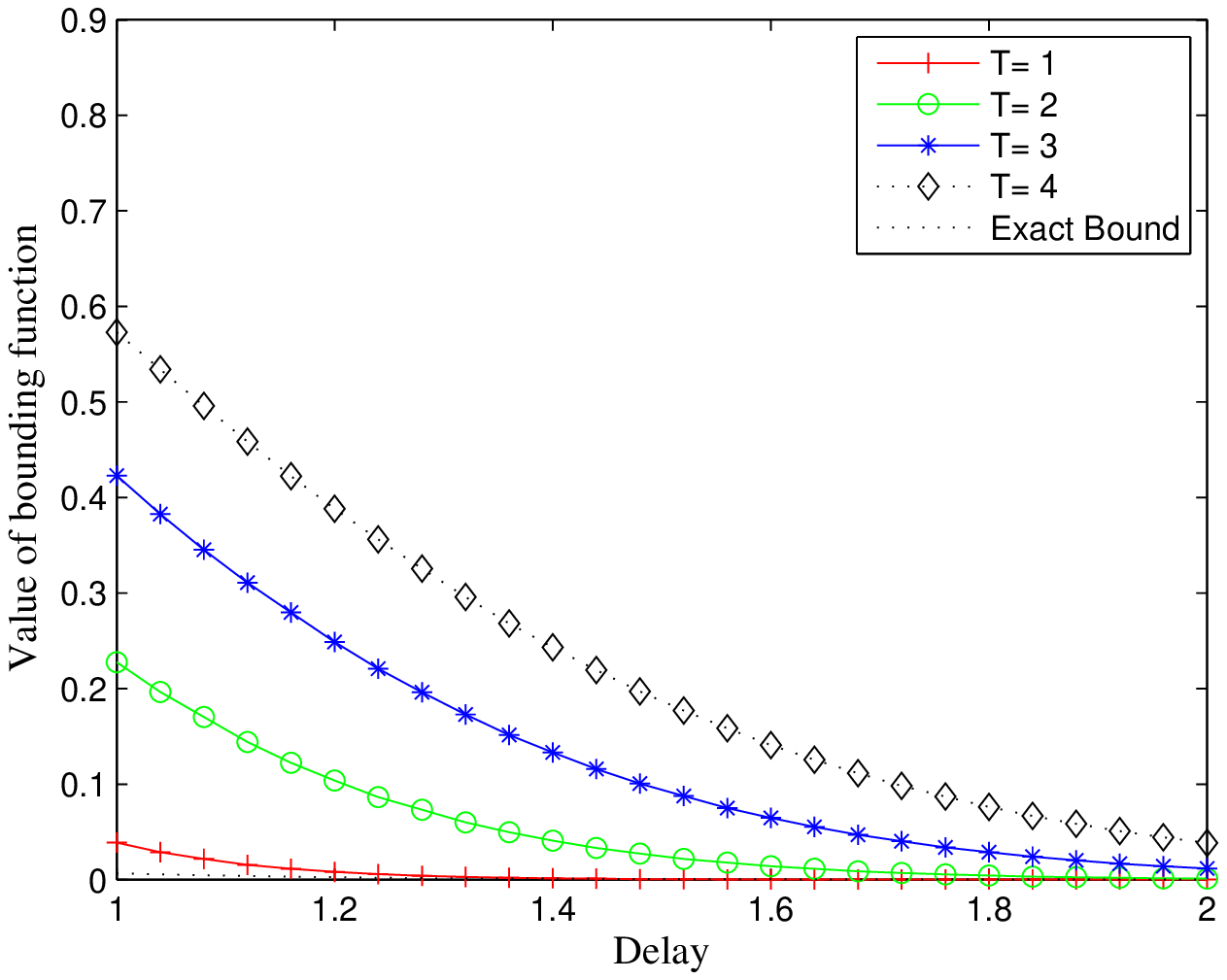}}
    \caption{\label{fig:theta1} The impact of $T$ on the bounding function ($\theta=0.1$)}
  \end{minipage}\hspace{2mm}
  \begin{minipage}{75 mm}
    \centerline{\includegraphics[width=60 mm]{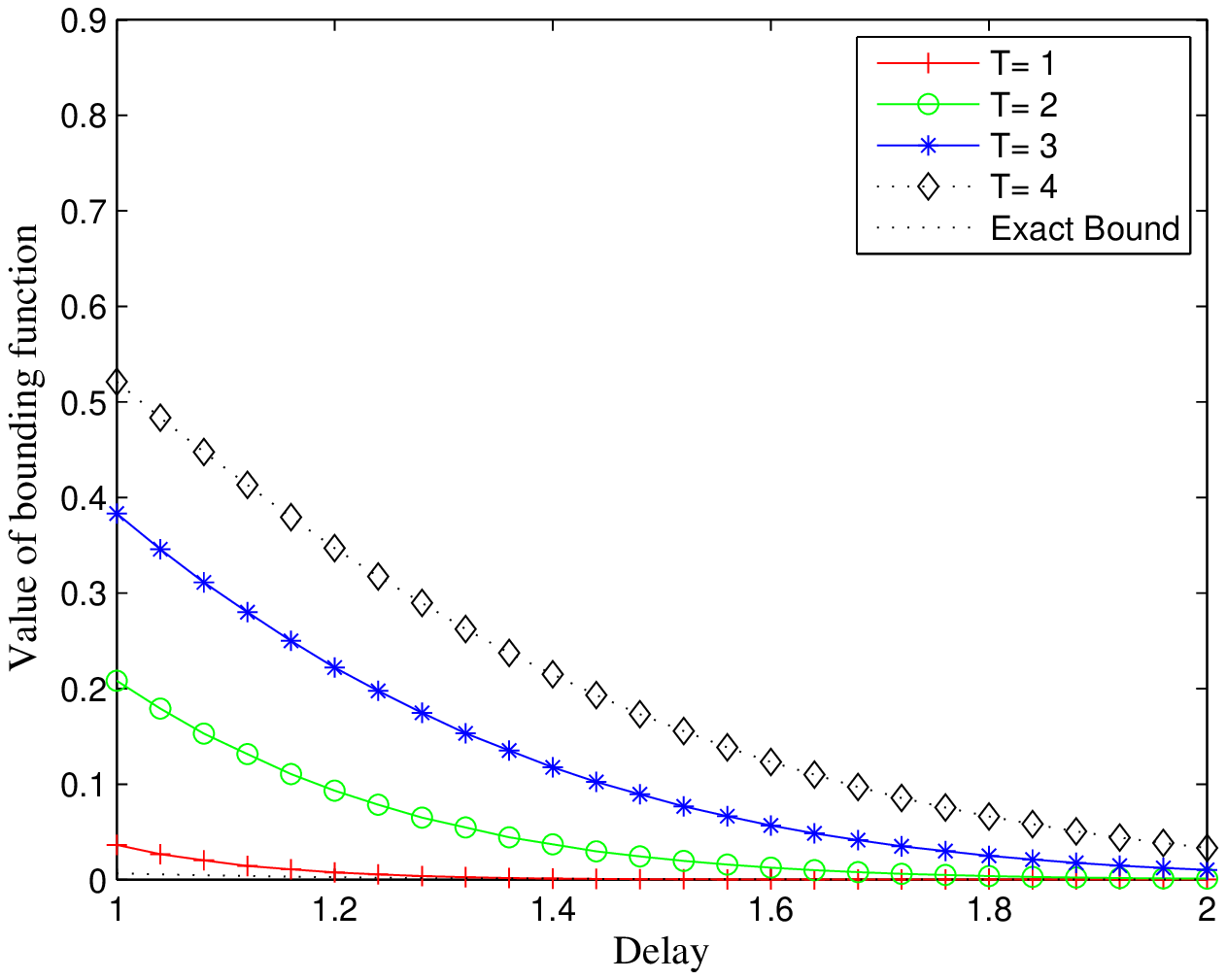}}
    \caption{\label{fig:theta4} The impact of $T$ on the bounding function ($\theta=0.8$)}
  \end{minipage}
\end{center}

\end{figure*}
%

    

Figs~\ref{fig:theta1} and~\ref{fig:theta4} show the numerical results on delay bounds for an $M/M/1$ queueing system, where the average traffic arrival rate $\lambda$ is equal to $20$ packets per second and the average service rate $\mu$ is equal to $25$ packets. The results demonstrate that the bounding function is very sensitive to the maximum time scale $T$. When $T$ is large, e.g, $4$ seconds, the bounding function becomes very loose. Comparing Fig.\ref{fig:theta1} and Fig.~\ref{fig:theta4}, we note that the impact of $\theta$ on the delay bounds is not significant, since the value of $\theta$ must be bounded by $\mu-\lambda$ so that the delay bound does not escape to infinity.
 
We can use traditional queueing analysis method to obtain the exact solution for the $M/M/1$ model. From~\cite{Kle}, the packet delay in the $M/M/1$ system,  $D(t)$, is an exponential random variable with mean $\frac{1}{\mu-\lambda}$, i.e., $Prob\{D(t) > x\} = e^{-(\mu-\lambda)}$. This is the exact solution, as it gives the exact distribution of packet delay.  From Figs~\ref{fig:theta1} and~\ref{fig:theta4}, it is easy to see that the exact bound obtained with the traditional queueing analysis method is much tighter than that with stochastic network calculus, if the stochastic traffic arrival and service models are not selected properly as above.
 
\subsection{Guidelines} 

From our analysis and numerical example, we summarize the following practical guidelines for model selection and transform in stochastic network calculus. 
\begin{itemize}
\item \textbf{Rule 1:} In general, the transform from a weak model (e.g., the $t.a.c.$ model ) to a strong model (e.g., the $v.b.c.$ model) will result in a looser bounding function. This problem cannot be solved by optimizing the $\theta$ value in the model transform.  If a strong model can be found directly with some methods like those in~\cite{Fry,Jiangbook,Kingman}, never rely on model transform to obtain the strong model from a weak model. 
\item \textbf{Rule 2:} When the bounding functions are time dependent, it is likely that the time-increasing problem will occur. The method of limiting the time with a maximum time scale is very tricky and should be clearly justified in each application. For instance, for traffic arrival curves, it could be set as the time until buffer overflows, or it could be set as a value that has been verified from real trace data. Nevertheless, as we have seen from the above example, the bounding functions are very sensitive to the maximum time scale. Never use a small maximum time scale value for the purpose of obtaining a tight performance bound without clear evidence demonstrating why the value is practical for the application in consideration. 
\item \textbf{Rule 3:} After model transform, another constraint should be posed on the selection of $\theta$, that is, the value of $\theta$ must make the system stable. Specifically, after the transform the value of $h(\alpha,\beta)$ must be bounded, where $\alpha$ and $\beta$ are the arrival curve and the service curve, respectively.     
\end{itemize}

Following the above rules, we can obtain tighter bounds shown in Figs.~\ref{fig:theta1NoIntegral} and~\ref{fig:theta2NoIntegral}, if we directly build a $v.b.c$ model using existing results from queueing theory~\cite{Fry,Kingman}, specifically, the result on the (virtual) queueing length distribution with time $t$.

\begin{figure*}[htb!]
\begin{center} \begin{minipage}{75 mm}
    \centerline{\includegraphics[width=60mm]{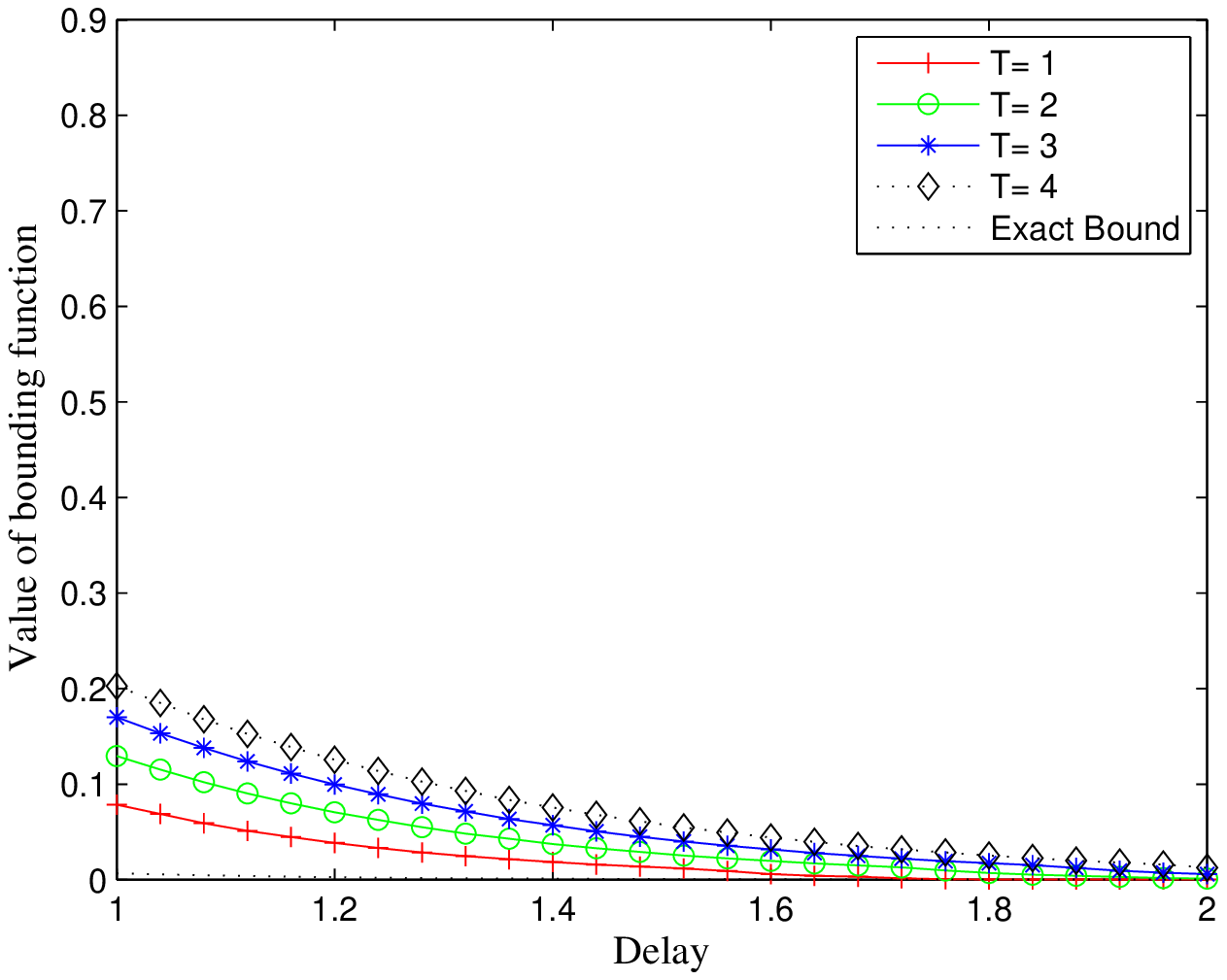}}
    \caption{\label{fig:theta1NoIntegral} Improved bounds ($\theta=0.1$)}
  \end{minipage}\hspace{2mm}
  \begin{minipage}{75 mm}
    \centerline{\includegraphics[width=60 mm]{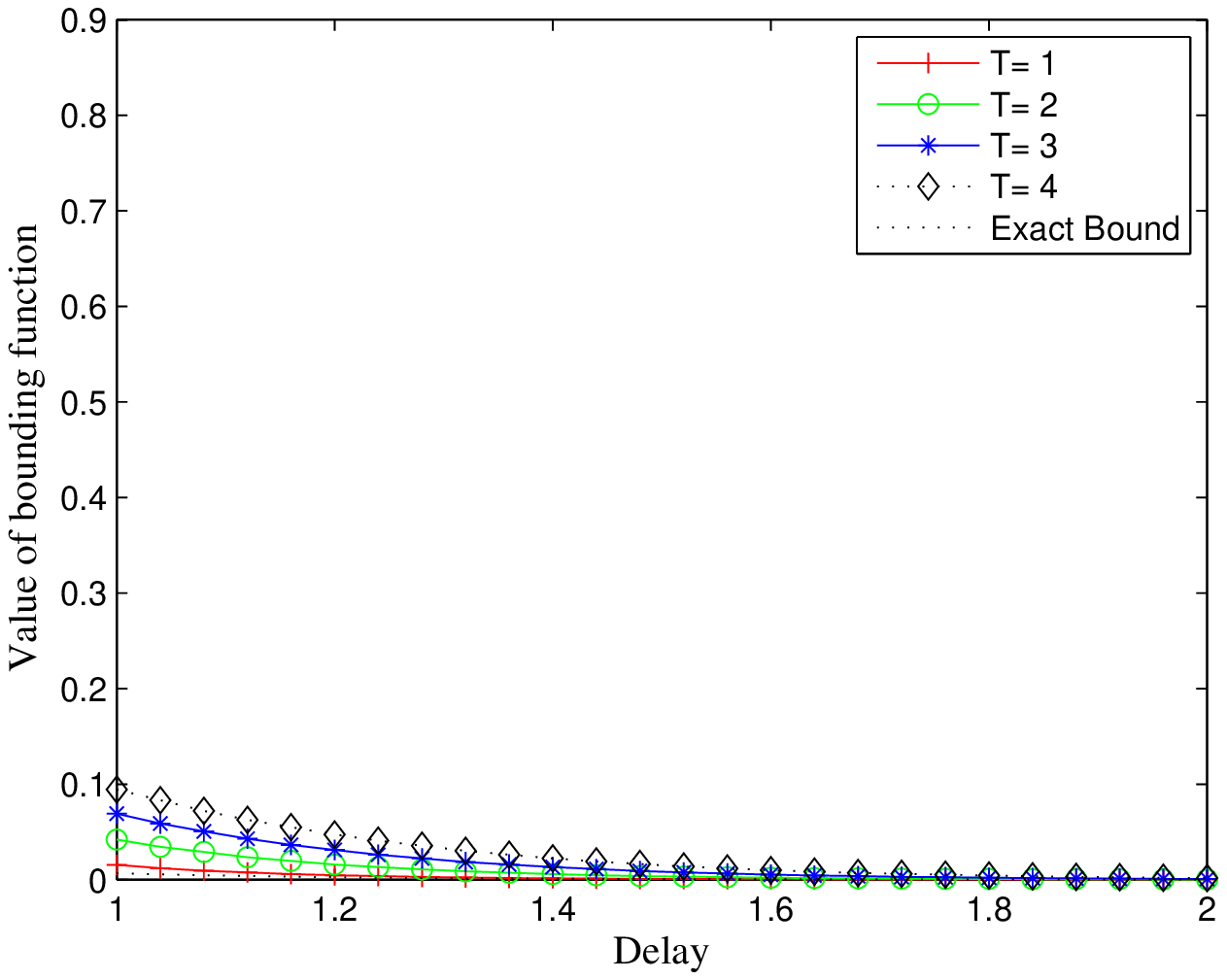}}
    \caption{\label{fig:theta2NoIntegral}Improved bounds ($\theta=0.8$)}
  \end{minipage}
\end{center}
\end{figure*}

\subsection{Further Discussion:  $\theta$-$m.b.c.$ Model And $\theta$-$s.c.$ Model} 
In the above guidelines, the most difficult part is to handle the time-increasing problem. Another method is proposed to avoid the time increase problem in~\cite{Jiangbook} by using a new traffic arrival model and a new service model, namely $\theta$-$m.b.c$ model and $\theta$-$s.c.$ model.  We only discuss the $\theta$-$m.b.c.$ model, since the same conclusion is applicable to the $\theta$-$s.c.$ model. 

\begin{definition}~\cite{Jiangbook} A flow is said to have a $\theta$-$m.b.c.$ stochastic arrival curve $\alpha \in \mathcal{F}$ with respect to $\theta$, with bounding function $f^\theta(x) \in \bar{\mathcal{F}}$, if for all $t\ge 0$ and all $x\ge 0$, there holds for some $\theta \ge 0$, 
\begin{equation}
Prob\{\sup_{0\le s\le t}[\sup_{0\le \mu \le s} A(\mu,s)-\alpha(s-\mu)-\theta \cdot (t-s)] > x \} \le f^\theta(x). 
\end{equation}
\end{definition}

Clearly, the $\theta$-$m.b.c.$ model is a scaling-up method, i.e., it raises the curve $\alpha$ to get a tighter bounding function, hopefully to obtain a bounding function that is independent of time $t$. It has been proved~\cite{Jiangbook} that if a flow has a $v.b.c$ stochastic arrival curve $\alpha$ with bounding function $f(x)\in \bar{\mathcal{G}}$, it has a $\theta$-$m.b.c$ stochastic arrival curve $\alpha_\theta$ with bounding function $f^\theta$, where for any $\theta >0$, $$\alpha_\theta(t) = \alpha(t) + \theta \cdot t,  f^\theta(x) = \left [f(x) + \frac{1}{\theta} \int_x^\infty f(y) dy\right]_1.$$

\begin{itemize}
\item \textbf{Rule 4:} Generally speaking, if we can model a flow with the $\theta$-$m.b.c.$ traffic arrival curve such that the bounding function is independent of time $t$, we can find for the flow a $v.b.c.$ traffic arrival model with a bounding function independent of time $t$ as well. Nevertheless, if the time-increasing problem exists in a weak model, model transform to a $\theta$-$m.b.c.$ model may not be very helpful.   
\end{itemize}
   
\section{Conclusion}

An useful analytical technique for performance evaluation, stochastic network calculus has not got the deserved fame and its application in practice is lagging far behind the theoretical development. \nop{It seems that the higher bar of background knowledge and the lack of guidelines in handling the tricky model selection and model transforms deter many from this field.} This should not be the case. This paper is to provide the guidance in the practical use of stochastic network calculus. Following the suggestions and the two-step approach in the paper should help the novice quickly grasp this useful technique. To conclude, we borrow Kleinrock's last words in his classical queueing theory book~\cite{Kle}: \textit{It now remains for you, the reader, to sharpen and apply the new set of tools. The world awaits and you must serve!}

\nop{In the future, we plan to build a series of stochastic models for common traffic arrivals and services, and a software package to calculate the stochastic performance bounds with these models. 

The basic requirements for such models directly usable for the second phase include: (1) The bounding function should not have the time-increasing problem; (2) They should be directly usable for a specific analysis without resorting to model transform. }

\section*{Acknowledgment}
This work was partially supported by the Natural Sciences and Engineering Research Council of Canada (NSERC) and the Japan Society for the Promotion of Science (JSPS) fellowship.   

\bibliographystyle{abbrv}
\bibliography{reference}
\small \baselineskip 9pt
\appendix
\section*{Proofs of Results} 
\label{app:A}

\textbf{Proof of Lemma~\ref{Lemma:PA}. } Given a Poisson arrival
process $A(t)$ with a rate $\lambda$, based on Chernoff
bounds~\cite{Mic}, we have for $\forall \theta \ge 0$,

\begin{equation}\label{eq1}
P\{ A(t) - \lambda t \ge x\}  = P\{ e^{\theta A(t)}  \ge e^{\theta
(\lambda t + x)} \}  \le \frac{{E(e^{\theta A(t)} )}}{{e^{\theta
(\lambda t + x)} }}.
\end{equation}
It is easy to see that $E(e^{\theta A(t)} ) = e^{\lambda t(e^\theta
- 1 )}$ (e.g., refer to Lemma 5.3 in~\cite{Mic}). Therefore,
(\ref{eq1}) is equivalent to
\begin{equation}\label{eq2}
P\{ A(t) - \lambda t \ge x\}  \le \frac{{e^{\lambda t(e^\theta  - 1
)} }}{{e^{\theta (\lambda t + x)} }}.
\end{equation}
To obtain the tightest bound of the above inequality, we calculate
\begin{equation*}\label{eq3}
\min (\frac{{e^{\lambda t(e^\theta  - 1 )} }}{{e^{\theta (\lambda t
+ x)} }}) = e^{\min (\lambda t(e^\theta  - 1 ) - \theta (\lambda t +
x))},
\end{equation*}
which can be obtained by calculating
\begin{equation}\label{eq4}
\frac{{d(\lambda t(e^\theta   - 1) - \theta (\lambda t +
x))}}{{d\theta }} = \lambda t \cdot e^\theta   - (\lambda t + x) =
0.
\end{equation}
From Equation (\ref{eq4}), we obtain
\begin{equation}\label{eq5}
\theta  = In\frac{{\lambda t + x}}{{\lambda t}}.
\end{equation}
Finally, the theorem is proved by applying (\ref{eq5}) into
(\ref{eq2}). \QED

\end{document}